\documentclass[11pt,letterpaper]{article}
\usepackage{amsmath,amssymb,mathtools,bm,color,enumitem,xspace}
\usepackage[pagebackref]{hyperref}
\usepackage[margin=1in]{geometry}
\usepackage[amsthm,amsmath,thmmarks]{ntheorem}
\usepackage[T1]{fontenc} 

\definecolor{hyper}{RGB}{0,0,88}
\hypersetup{pdfborder={0 0 0},colorlinks=true,linkcolor=hyper,citecolor=hyper,filecolor=hyper,urlcolor=hyper}

\renewcommand{\backref}[1]{}

\renewcommand{\backrefalt}[4]{%
\ifcase #1 %
\or
[p.\ #2]%
\else
[pp.\ #2]%
\fi}

\newtheorem{theorem}{Theorem}

\newtheorem{lemma}{Lemma}
\newtheorem{claim}{Claim}

\newtheorem{fact}{Fact}
\newtheorem{observation}{Observation}

\newtheorem{open}{Open Question}

\newcommand{\customthmname}{}
\newtheorem*{customthm*}{\customthmname}

\newenvironment{mylist}[1]{\begin{itemize}\setlength{\itemsep}{#1pt}\setlength{\parsep}{0pt}\setlength{\parskip}{0pt}}{\end{itemize}}

\newcommand{\Id}{\textsc{Id}}
\newcommand{\Or}{\textsc{Or}}
\newcommand{\GapOr}{\textsc{GapOr}}
\newcommand{\Xor}{\textsc{Xor}}
\newcommand{\Maj}{\textsc{Maj}}
\newcommand{\GapMaj}{\textsc{GapMaj}}
\newcommand{\Which}{\textsc{Which}}
\newcommand{\Omb}{\textsc{Omb}}

\newcommand{\dt}{{{\text{\upshape\sffamily dt}}\xspace}}

\renewcommand{\P}{{{\text{\upshape\sffamily P}}\xspace}}
\newcommand{\RP}{{{\text{\upshape\sffamily RP}}\xspace}}

\newcommand{\BPP}{{{\text{\upshape\sffamily BPP}}\xspace}}
\newcommand{\BPPbar}{{\overline{\text{\upshape\sffamily BPP}}\xspace}}
\newcommand{\tWAPP}{{{\text{\upshape\sffamily 2WAPP}}\xspace}}
\newcommand{\WAPP}{{{\text{\upshape\sffamily WAPP}}\xspace}}
\newcommand{\AWPP}{{{\text{\upshape\sffamily AWPP}}\xspace}}
\newcommand{\PostBPP}{{{\text{\upshape\sffamily PostBPP}}\xspace}}
\newcommand{\PP}{{{\text{\upshape\sffamily PP}}\xspace}}

\newcommand{\Rbar}{\overline{{\text{\upshape R}}}\xspace}

\renewcommand{\Pr}{\mathbb{P}}
\newcommand{\E}{\mathbb{E}}

\newcommand{\calD}{{\mathcal D}}
\newcommand{\calF}{{\mathcal F}}
\newcommand{\calG}{{\mathcal G}}

\renewcommand{\mid}{\,|\,}
\newcommand{\bigmid}{\,\big|\,}

\renewcommand{\i}{\textit{(i)}\xspace}
\newcommand{\ii}{\textit{(ii)}\xspace}

\begin{document}

\title{When Is Amplification Necessary for Composition\\in Randomized Query Complexity?}

\author{Shalev Ben-David\footnote{University of Waterloo. \texttt{shalev.b@uwaterloo.ca}}
\and Mika G\"o\"os\footnote{Stanford University. \texttt{goos@stanford.edu}}
\and Robin Kothari\footnote{Microsoft Quantum and Microsoft Research. \texttt{robin.kothari@microsoft.com}}
\and Thomas Watson\footnote{University of Memphis. \texttt{Thomas.Watson@memphis.edu}}
}

\maketitle

\begin{abstract}
Suppose we have randomized decision trees for an outer function $f$ and an inner function $g$. The natural approach for obtaining a randomized decision tree for the composed function $(f\circ g^n)(x^1,\ldots,x^n)=f(g(x^1),\ldots,g(x^n))$ involves amplifying the success probability of the decision tree for $g$, so that a union bound can be used to bound the error probability over all the coordinates. The amplification introduces a logarithmic factor cost overhead. We study the question: When is this log factor necessary? We show that when the outer function is parity or majority, the log factor can be necessary, even for models that are more powerful than plain randomized decision trees. Our results are related to, but qualitatively strengthen in various ways, known results about decision trees with noisy inputs.
\end{abstract}


\section{Introduction} \label{sec:intro}

A deterministic decision tree for computing a partial function $f\colon\{0,1\}^n\to Z$ is a binary tree where each internal node is labeled with an index from $[n]$ and each leaf is labeled with an output value from $Z$. On input $x\in\{0,1\}^n$, the computation follows a root-to-leaf path where at a node labeled with index $i$, the value of $x_i$ is queried and the path goes to the left child if $x_i=0$ and to the right child if $x_i=1$. The leaf reached on input $x$ must be labeled with the value $f(x)$ (if the latter is defined). The cost of the decision tree is its depth, i.e., the maximum number of queries it makes over all inputs. The deterministic query complexity of $f$ is the minimum cost of any deterministic decision tree that computes $f$. We will consider several more general models of decision trees (randomized, etc.), so we repurpose traditional complexity class notation to refer to the various associated query complexity measures. Since $\P$ is the traditional complexity class corresponding to deterministic computation, we let $\P(f)$ denote the deterministic query complexity of $f$. (Some of the recent literature uses the notation $\P^\dt(f)$, but this paper deals exclusively with decision trees, so we drop the $\dt$ superscript.)

A randomized decision tree is a probability distribution over deterministic decision trees. Computing $f$ with error $\varepsilon$ means that for every input $x$ (for which $f(x)$ is defined), the probability that the output is not $f(x)$ is at most $\varepsilon$. The cost of a randomized decision tree is the maximum depth of all the deterministic trees in its support. The randomized query complexity $\BPP_\varepsilon(f)$ is the minimum cost of any randomized decision tree that computes $f$ with error $\varepsilon$. When we write $\BPP(f)$ with no $\varepsilon$ specified, we mean $\varepsilon=1/3$. A basic fact about randomized computation is that the success probability can be amplified, with a multiplicative overhead in cost, by running several independent trials and taking the majority vote of the outputs: $\BPP_\varepsilon(f)\le O(\BPP(f)\cdot\log(1/\varepsilon))$. See \cite{buhrman02complexity} for a survey of classic results on query complexity.

If $f\colon\{0,1\}^n\to Z$ and $g\colon\{0,1\}^m\to\{0,1\}$ are two partial functions, their composition is $f\circ g^n\colon(\{0,1\}^m)^n\to Z$ where $(f\circ g^n)(x^1,\ldots,x^n)\coloneqq f(g(x^1),\ldots,g(x^n))$ (which is defined iff $g(x^i)$ is defined for all $i$ and $f(g(x^1),\ldots,g(x^n))$ is defined). How does the randomized query complexity of $f\circ g^n$ depend on the randomized query complexities of $f$ and $g$? A simple observation is that to design a randomized decision tree for $f\circ g^n$, we can take a $1/6$-error randomized decision tree for $f$ and replace each query---say to the $i^\text{th}$ input bit of $f$---with a $1/6n$-error randomized decision tree for evaluating $g(x^i)$. By a union bound, with probability at least $5/6$ all of the (at most $n$) evaluations of $g$ return the correct answer, and so with probability at least $2/3$ the final evaluation of $f$ is also correct. Since $\BPP_{1/6n}(g)\le O(\BPP_{1/n}(g))$, we can write this upper bound as
\begin{equation} \label{eq:simple-bound}
\BPP(f\circ g^n)~\le~O(\BPP(f)\cdot\BPP_{1/n}(g))~\le~O(\BPP(f)\cdot\BPP(g)\cdot\log n).
\end{equation}
\textit{When is this tight?} It will take some effort to suitably formulate this question. We begin by reviewing known related results.


\subsection{When is amplification necessary?} \label{sec:intro:formulate}

As for general lower bounds (that hold for all $f$ and $g$), much work has gone into proving lower bounds on $\BPP(f\circ g^n)$ in terms of complexity measures of $f$ and $g$ that are defined using models more powerful than plain randomized query complexity \cite{goos16composition,anshu17composition,bendavid18randomized,bassilakis20power,bendavid20tight}. In terms of just $\BPP(f)$ and $\BPP(g)$, the state-of-the-art is that $\BPP(f\circ g^n)\ge\Omega(\BPP(f)\cdot\sqrt{\BPP(g)})$ for all $f$ and $g$ \cite{gavinsky19composition}. Furthermore, it is known that the latter bound is sometimes tight: There exist partial boolean functions $f$ and $g$ such that $\BPP(f\circ g^n)\le\widetilde{O}(\BPP(f)\cdot\sqrt{\BPP(g)})$ and $\BPP(f),\BPP(g)\ge\omega(1)$ \cite{gavinsky19composition,bendavid20tight}. Thus \eqref{eq:simple-bound} is far from being \emph{always} tight, even without worrying about the need for amplification. However, it remains plausible that $\BPP(f\circ g^n)\ge\Omega(\BPP(f)\cdot\BPP(g))$ holds for all \emph{total} $f$ and all partial $g$. We take this as a working conjecture in this paper. This conjecture has been confirmed for some specific outer functions $f$, such as the identity function $\Id\colon\{0,1\}^n\to\{0,1\}^n$ \cite{jain10optimal} (this is called a ``direct sum'' result) and the boolean functions $\Or$, $\Xor$ (parity), and $\Maj$ (majority) \cite{goos18randomized}. These results, however, do not address the need for amplification in the upper bound \eqref{eq:simple-bound}. To formulate our question of whether \eqref{eq:simple-bound} is tight, a first draft could be: \[\text{\textbf{Question A}, with respect to a particular $f$:~~~Is \eqref{eq:simple-bound} tight for \emph{all} partial functions $g$?}\] This is not quite a fair question, for at least two reasons:
\begin{mylist}{2}
\item Regarding the first inequality in \eqref{eq:simple-bound}: The simple upper bound actually shows $\BPP(f\circ g^n)\le O(\BPP(f)\cdot\BPP_{1/\BPP(f)}(g))$ (the union bound is only over queries that take place, not over all possible queries). So for simplicity, let us restrict our attention to $f$ satisfying $\BPP(f)\ge\Omega(n)$, which is the case for $\Id$, $\Or$, $\Xor$, and $\Maj$.
\item Regarding the second inequality in \eqref{eq:simple-bound}: Some functions $g$ satisfy $\BPP_{1/n}(g)\le o(\BPP(g)\cdot\log n)$ (e.g., if $\P(g)\le O(\BPP(g))$). So for simplicity, let us restrict our attention to $g$ satisfying $\BPP_{1/n}(g)\ge\Omega(\BPP(g)\cdot\log n)$, which (as we show later) is the case for two partial functions $\GapOr$ and $\GapMaj$ defined as follows ($|x|$ denotes the Hamming weight of $x\in\{0,1\}^m$): \[\GapOr(x)\coloneqq\begin{cases}0&\text{if $|x|=0$}\\1&\text{if $|x|=m/2$}\end{cases}\quad\text{and}\quad\GapMaj(x)\coloneqq\begin{cases}0&\text{if $|x|=m/3$}\\1&\text{if $|x|=2m/3$}\end{cases}.\]
\end{mylist}
Thus, a better formulation of Question A would be: Assuming $\BPP(f)\ge\Omega(n)$, is \eqref{eq:simple-bound} tight for all partial $g$ satisfying $\BPP_{1/n}(g)\ge\Omega(\BPP(g)\cdot\log n)$? Even with these caveats, the answer is always ``no.'' It will be instructive to examine a counterexample. Let $\Which\colon\{0,1\}^2\to\{0,1\}$ be the partial function such that $\Which(y)$ indicates the location of the unique $1$ in $y$, under the promise that $|y|=1$. Then $g=\Which\circ\GapOr^2$ takes an input of length $2m$ with the promise that there are exactly $m/2$ many $1$s, either all in the left half or all in the right half, and outputs which half has the $1$s. It turns out $\BPP(g)\le O(1)$ and $\BPP_{1/n}(g)\ge\Omega(\log n)$ provided $m\ge\log n$ (for similar reasons as $\GapOr$ itself) and yet $\BPP(f\circ g^n)\le O(\BPP(f))$ for all $f$: To compute $f\circ g^n$, we can run an optimal randomized decision tree for $f$ and whenever it queries $g(x^i)$, we repeatedly query uniformly random bit positions of $x^i$ until we find a $1$ (so the value of $g(x^i)$ is determined by which half we found a $1$ in). This has the same error probability as the randomized decision tree for $f$, and the total number of queries to the bits of $(x^1,\ldots,x^n)$ is $O(\BPP(f))$ in expectation, because for each $i$ it takes $O(1)$ queries in expectation to locate a $1$ in $x^i$. By Markov's inequality, with high constant probability this halts after only $O(\BPP(f))$ total queries. Thus by aborting the computation if it attempts to make too many queries, we obtain a randomized decision tree for $f\circ g^n$ that always makes $O(\BPP(f))$ queries, with only a small hit in the error probability.

Blais and Brody \cite{blais19optimal} adjust the statement of Question A so the answer becomes ``yes'' in the case $f=\Id$. Specifically, they weaken the right-hand side in such a way that the above counterexample is ruled out. Defining\footnote{\cite{blais19optimal} used the notation $\Rbar$ instead of $\BPPbar$.} $\BPPbar_\varepsilon(g)$ similarly to $\BPP_\varepsilon(g)$ but where the cost of a randomized decision tree is the maximum over all inputs (on which $g$ is defined) of the expected number of queries, we now have $\BPPbar_{1/n}(g)\le\BPPbar_0(g)\le O(1)$ for the $g$ from the counterexample. The theorem from \cite{blais19optimal} is $\BPP(f\circ g^n)\ge\Omega(\BPP(f)\cdot\BPPbar_{1/n}(g))$ when $f=\Id$, in other words, $\BPP(g^n)=\Omega(n\cdot\BPPbar_{1/n}(g))$ (a ``strong direct sum'' result). \cite{blais19optimal} also explicitly asked whether similar results hold for other functions $f$. The corresponding conjecture for $f=\Or$ is false (as we note below) while for $f=\Xor$ and $f=\Maj$ it remains open.

To make progress, we step back and ask a seemingly more innocuous version of the question: \[\text{\textbf{Question B}, with respect to a particular $f$:~~~Is \eqref{eq:simple-bound} tight for \emph{some} partial function $g$?}\] It turns out the answer is ``no'' for $f=\Or$ and is ``yes'' for both $f=\Xor$ and $f=\Maj$.


\subsection{Decision trees with noisy inputs} \label{sec:intro:or}

Question B is related to ``query complexity with noisy inputs'' (introduced in \cite{feige94computing}), so let us review the latter model: When input bit $y_i$ is queried, the wrong bit value is returned to the decision tree with some probability $\le 1/3$ (and the correct value of $y_i$ is returned with the remaining probability). The ``noise events'' are independent across all queries, including multiple queries to the same input bit. Now the adversary gets to pick not only the input, but also the ``noise probabilities.'' \cite{feige94computing} distinguishes between two extreme possibilites: A static adversary has a single common noise probability for all queries, while a dynamic adversary can choose a different noise probability for each node in the decision tree. In this paper we make a reasonable compromise: The adversary gets to choose a tuple of noise probabilities $(\nu_1,\ldots,\nu_n)$, and each query to $y_i$ returns $1-y_i$ with probability exactly $\nu_i$. When a randomized decision tree computes $f$ with error probability $\varepsilon$, that means for every input $y\in\{0,1\}^n$ and every noise probability tuple $(\nu_1,\ldots,\nu_n)$ (with $\nu_i\le 1/3$ for each $i$), the output is $f(y)$ with probability $\ge 1-\varepsilon$ over the random noise and randomness of the decision tree. We invent the notation $\BPP^*(f)$ for the minimum cost of any randomized decision tree that computes $f$ on noisy inputs, with error probability $1/3$. We have $\BPP^*(f)\le O(\BPP(f)\cdot\log n)\le O(n\log n)$ by repeating each query $O(\log n)$ times and taking the majority vote (to drive the noise probabilities down to $o(1/n)$), and using a union bound to absorb the noise probabilities into the error probability. The connection with composition is that $\BPP(f\circ g^n)\le\BPP^*(f)\cdot\BPP(g)$, because to design a randomized decision tree for $f\circ g^n$, we can take a $1/3$-error randomized decision tree for $f$ with noisy inputs, and replace each query---say to $y_i$---with a $1/3$-error randomized decision tree for evaluating $g(x^i)$.

There is a similar connection for $1$-sided error and $1$-sided noise. When a randomized decision tree has $1$-sided error $\varepsilon$, that means on $0$-inputs the output is wrong with probability $0$, and on $1$-inputs the output is wrong with probability at most $\varepsilon$. We let $\RP(g)$ denote the minimum cost of any randomized decision tree that computes $g$ with $1$-sided error $1/2$. Similarly, $1$-sided noise means that when input bit $y_i$ is queried, if the actual value is $y_i=0$ then $1$ is returned with probability $0$, and if the actual value is $y_i=1$ then $0$ is returned with probability $\nu_i\le 1/2$. We invent the notation $\BPP^\dagger(f)$ for the minimum cost of any randomized decision tree that computes $f$ on $1$-sided noisy inputs, with $2$-sided error probability $1/3$. We have $\BPP(f)\le\BPP^\dagger(f)\le\BPP^*(f)$. The connection $\BPP(f\circ g^n)\le\BPP^\dagger(f)\cdot\RP(g)$ holds like in the $2$-sided noise setting. We officially record these observations:

\begin{observation} \label{obs:noisy-comp}
For all $f$ and $g$, \[\BPP(f\circ g^n)~\le~\BPP^*(f)\cdot\BPP(g)\quad\text{and}\quad\BPP(f\circ g^n)~\le~\BPP^\dagger(f)\cdot\RP(g).\]
\end{observation}

The upshot is that noisy upper bounds imply composition upper bounds, and composition lower bounds imply noisy lower bounds. There are many proofs of the result $\BPP^*(\Or)\le O(n)$ \cite{feige94computing,kenyon94boolean,newman09computing,goyal10rounds}:

\begin{theorem}[$\Or$ never necessitates amplification] \label{thm:or}
$\BPP^*(\Or)\le O(n)$ and thus for every partial function $g$, \[\BPP(\Or\circ g^n)~\le~O(n\cdot\BPP(g)).\]
\end{theorem}

\autoref{thm:or} is not new, but in \autoref{sec:or} we provide a particularly clean and elementary proof (related to, but more streamlined than, the proof in \cite{kenyon94boolean}). We mention that the proof straightforwardly generalizes to some other functions $f$, such as ``odd-max-bit'': $\Omb(y)=1$ iff the highest index of any $1$ in $y$ is odd.

We turn our attention to lower bounds. Various special-purpose techniques have been developed for proving query complexity lower bounds in the noisy setting \cite{feige94computing,evans98average,dutta08lower,goyal10rounds}. However, a conceptual consequence of \autoref{obs:noisy-comp} is that special-purpose techniques are not generally necessary: We can just use techniques for lower bounding plain (non-noisy) randomized query complexity, applied to composed functions.


\subsection{Lower bound for parity} \label{sec:intro:xor}

\cite{feige94computing} proved that $\BPP^*(\Xor)$ and $\BPP^*(\Maj)$ are $\Omega(n\log n)$. Although apparently not recorded in the literature, it is possible to generalize this result to show $\BPP^\dagger(\Xor)$ and $\BPP^\dagger(\Maj)$ are $\Omega(n\log n)$. However, we prove results even stronger than that, using the composition paradigm. Our results involve query complexity models that are more powerful than $\BPP$, and even more powerful than the $\BPPbar$ model from \cite{blais19optimal}. This follows a theme from a lot of prior work: Since $\BPP$ query complexity is rather subtle, we can make progress by studying related models that are somewhat more ``well-behaved.''

\begin{mylist}{2}
\item As observed in \cite{blais19optimal}, the $\BPPbar$ model is equivalent to one where the cost is the worst-case (rather than expected) number of queries, and a randomized decision tree is allowed to abort (i.e., output a special symbol $\bot$) with at most a small constant probability, and the output should be correct with high probability conditioned on not aborting.
\item If we strengthen the above model by allowing the non-abort probability to be arbitrarily close to $0$ (rather than close to $1$), but require that the non-abort probabilities are approximately the same for all inputs (within some factor close to $1$), the resulting model has been called $\tWAPP$ (``$2$-sided weak almost-wide $\PP$'') \cite{goos16rectangles,goos18randomized}. The ``$1$-sided'' version $\WAPP$, defined later, will be relevant to us.
\item If we further strengthen the model by allowing the non-abort probabilities to be completely unrelated for different inputs (and still arbitrarily close to $0$), the resulting model has been called $\PostBPP$ (``$\BPP$ with post-selection'') \cite{goos16rectangles,cade18post}.
\end{mylist}

We first consider the last of these models. $\PostBPP_\varepsilon(f)$ is the minimum cost of any randomized decision tree such that on every input $x$ (for which $f(x)$ is defined), the probability of outputting $\bot$ is $<1$, and the probability of outputting $f(x)$ is $\ge 1-\varepsilon$ conditioned on not outputting $\bot$. Trivially, $\PostBPP(f)\le\BPP(f)$. In fact, the $\PostBPP$ model is much more powerful than plain randomized query complexity; for example (noted in \cite{goos16rectangles}) it can efficiently compute the aforementioned odd-max-bit function: $\PostBPP(\Omb)\le 1$.

For the noisy input setting, $\PostBPP^*$ and $\PostBPP^\dagger$ are defined in the natural way, and $\PostBPP(f\circ g^n)\le\PostBPP^*(f)\cdot\BPP(g)$ and $\PostBPP(f\circ g^n)\le\PostBPP^\dagger(f)\cdot\RP(g)$ hold like in \autoref{obs:noisy-comp}.

In \autoref{sec:xor} we prove something qualitatively much stronger than $\BPP^*(\Xor)\ge\Omega(n\log n)$:

\begin{theorem}[$\Xor$ sometimes necessitates amplification] \label{thm:xor}
For some partial function $g$,\\ namely $g=\GapMaj$ with $m\ge\log n$, \[\PostBPP(\Xor\circ g^n)~\ge~\Omega(n\cdot\BPP_{1/n}(g))~\ge~\Omega(n\log n\cdot\BPP(g)).\] In particular, $\PostBPP^*(\Xor)\ge\Omega(n\log n)$.
\end{theorem}

Let us compare \autoref{thm:xor} to two previous results.
\begin{mylist}{2}
\item \cite{evans98average} proved that $\BPPbar^*(\Xor)\ge\Omega(n\log n)$ and that this lower bound holds even in the average-case setting (i.e., $\Omega(n\log n)$ queries are needed in expectation to succeed with high probability over a uniformly random input, random noise, and randomness of the decision tree). Our proof of \autoref{thm:xor} is simpler than the proof in \cite{evans98average} (though both proofs have a Fourier flavor), it also works in the average-case setting, and it yields a stronger result since the model is $\PostBPP$ instead of just $\BPPbar$ (and the lower bound holds for composition rather than just noisy inputs). \cite{dutta08lower} presented a different simplified proof of the result from \cite{evans98average}, but that proof does not generalize to $\PostBPP^*$.
\item Our proof of \autoref{thm:xor} shows something analogous, but incomparable, to the strong direct sum from \cite{blais19optimal}. As we explain in \autoref{sec:xor}, our proof shows that $\PostBPP(\Xor\circ g^n)\ge\Omega(n\cdot\PostBPP_{1/n}(g))$ holds for \emph{all} $g$ (thus addressing a version of our Question A). Compared to the \cite{blais19optimal} result that $\BPPbar(\Id\circ g^n)\ge\Omega(n\cdot\BPPbar_{1/n}(g))$ for all $g$, our result has the advantages of working for $f=\Xor$ rather than $f=\Id$ and yielding a qualitatively stronger lower bound ($\PostBPP$ rather than $\BPPbar$ on the left side), but the disadvantage of also requiring the qualitatively stronger type of lower bound on $g$. Our result shows that if amplifying $g$ requires a log factor in a very strong sense (even $\PostBPP$-type decision trees cannot avoid the log factor), then that log factor will be necessary when composing $\Xor$ with $g$.
\end{mylist}


\subsection{Lower bound for majority} \label{sec:intro:maj}

Our main result strengthens the bound $\BPP^*(\Maj)\ge\Omega(n\log n)$ from \cite{feige94computing}, mainly by holding for the stronger model $\WAPP$ (rather than just $\BPP$), but also by directly handling $1$-sided noise and by holding for composition rather than just noisy inputs.

$\WAPP_\varepsilon(f)$ is the minimum cost of any randomized decision tree such that for some $t>0$, on input $x$ the probability of outputting $1$ is in the range $[(1-\varepsilon)t,t]$ if $f(x)=1$, and in the range $[0,\varepsilon t]$ if $f(x)=0$. The $\varepsilon$ subscript should always be specified, because unlike $\BPP$ and $\PostBPP$, $\WAPP$ is not amenable to efficient amplification of the error parameter $\varepsilon$ \cite{goos16rectangles}. For every constant $0<\varepsilon<1/2$, we have $\PostBPP(f)\le O(\WAPP_\varepsilon(f))\le O(\BPP(f))$.

$\WAPP$-type query complexity has several aliases, such as ``approximate conical junta degree'' and ``approximate query complexity in expectation,'' and it has recently played a central role in various randomized query (and communication) complexity lower bounds \cite{kaniewski15query,goos16rectangles,goos16composition,goos18randomized}. One can think of $\WAPP$ as a nonnegative version of approximate polynomial degree (which corresponds to the class $\AWPP$); in other words, it is a classical analogue of the polynomial method used to lower bound quantum algorithms.

For the noisy input setting, $\WAPP^*$ and $\WAPP^\dagger$ are defined in the natural way, and $\WAPP_\varepsilon(f\circ g^n)\le\WAPP_\varepsilon^*(f)\cdot\BPP(g)$ and $\WAPP_\varepsilon(f\circ g^n)\le\WAPP_\varepsilon^\dagger(f)\cdot\RP(g)$ hold like in \autoref{obs:noisy-comp}. We prove the following theorem, which shows that $\WAPP$ sometimes requires amplification, even in the one-sided noise setting.

\begin{theorem}[$\Maj$ sometimes necessitates amplification] \label{thm:maj}
For some partial function $g$,\\ namely $g=\GapOr$ with $m\ge\log n$, and some constant $\varepsilon>0$, \[\WAPP_\varepsilon(\Maj\circ g^n)~\ge~\Omega(n\cdot\BPP_{1/n}(g))~\ge~\Omega(n\log n\cdot\RP(g)).\] In particular, $\WAPP_\varepsilon^\dagger(\Maj)\ge\Omega(n\log n)$.
\end{theorem}

This theorem should be contrasted with the work of Sherstov about making polynomials robust to noise \cite{sherstov13making}. In that work, Sherstov showed that approximate polynomial degree never requires a log factor in the noisy input setting, nor in composition. That is to say, he improved the simple bound $\AWPP^*(f)\le O(\AWPP(f)\cdot\log n)$ to $\AWPP^*(f)\le O(\AWPP(f))$ for all Boolean functions $f$, and showed $\AWPP(f\circ g^n)\le O(\AWPP(f)\cdot \AWPP(g))$. In contrast, for conical juntas (nonnegative linear combinations of conjunctions), \autoref{thm:maj} shows that in a strong sense, the simple bound $\WAPP_\epsilon^*(f)\le O(\WAPP_\delta(f)\cdot\log n)$ (for all constants $0<\delta<\varepsilon<1/2$ and total Boolean functions $f$) cannot be improved: $\WAPP_\varepsilon^\dagger(f)\ge\Omega(\WAPP_0(f)\cdot\log n)$ for some constant $\varepsilon$ and some total $f$, namely $f=\Maj$. Thus unlike polynomials, conical juntas cannot be made robust to noise.

Our proof of \autoref{thm:maj} (in \autoref{sec:maj}) introduces some technical ideas that may be useful for other randomized query complexity lower bounds.

By a simple reduction, \autoref{thm:maj} for $g=\GapOr$ implies the same for $g=\GapMaj$ (with $\BPP(g)=1$ instead of $\RP(g)=1$ at the end of the statement), but we do not know of a simpler direct proof for the latter result. \autoref{thm:maj} cannot be strengthened to have $\PostBPP$ in place of $\WAPP$, because $\PostBPP(\Maj\circ\GapMaj^n)\le O(n)$. However, \autoref{thm:maj} does hold with $\Xor$ in place of $\Maj$, by the same proof.


\section{Proof of \protect{\autoref{thm:xor}}: \texorpdfstring{$\Xor$}{Xor} sometimes necessitates amplification} \label{sec:xor}

We first discuss a standard technique for proving randomized query complexity lower bounds, which will be useful in the proof of \autoref{thm:xor}. For any conjunction $C\colon\{0,1\}^k\to\{0,1\}$ and distribution $\calD$ over $\{0,1\}^k$, we write $C(\calD)\coloneqq\E_{x\sim\calD}[C(x)]=\Pr_{x\sim\calD}[C(x)=1]$. The number of literals in a conjunction is called its width.

\begin{fact} \label{fact:postbpp}
Let $h\colon\{0,1\}^k\to\{0,1\}$ be a partial function, and for each $z\in\{0,1\}$ let $\calD_z$ be a distribution over $h^{-1}(z)$. Then for every $\varepsilon$ there exist a conjunction $C$ of width $\PostBPP_\varepsilon(h)$ and a $z\in\{0,1\}$ such that $\varepsilon\cdot C(\calD_z)\ge(1-\varepsilon)\cdot C(\calD_{1-z})$ and $C(\calD_z)>0$.
\end{fact}

\begin{proof}
Abbreviate $\PostBPP_\varepsilon(h)$ as $r$. Fix a randomized decision tree of cost $r$ computing $h$ with error $\varepsilon$ conditioned on not aborting, and assume w.l.o.g.\ that for each outcome of the randomness, the corresponding deterministic tree is a perfect tree with $2^r$ leaves, all at depth $r$. Consider the probability space where we sample input $x$ from the mixture $\frac{1}{2}\calD_0+\frac{1}{2}\calD_1$, sample a deterministic decision tree $T$ as an outcome of the randomized decision tree, and sample a uniformly random leaf $\ell$ of $T$. Let $A$ be the indicator random variable for the event that $\ell$ is the leaf reached by $T(x)$ and its label is $h(x)$. Let $B$ be the indicator random variable for the event that $\ell$ is the leaf reached by $T(x)$ and its label is $1-h(x)$. Conditioned on any particular $x$ and $T$, the probability that $\ell$ is the leaf reached by $T(x)$ is $2^{-r}$. Thus conditioned on any particular $x$, if the non-abort probability is $t_x>0$ then $\E[A\mid x]\ge 2^{-r}t_x(1-\varepsilon)$ and $\E[B\mid x]\le 2^{-r}t_x\varepsilon$ and thus $\varepsilon\cdot\E[A\mid x]-(1-\varepsilon)\cdot\E[B\mid x]\ge 0$. Over the whole probability space, we have $\varepsilon\cdot\E[A]-(1-\varepsilon)\cdot\E[B]\ge 0$, so by linearity the same must hold conditioned on some particular $T$ and $\ell$ with $\E[A\mid T,\ell]>0$. Let $C$ be the conjunction of width $r$ such that $C(x)=1$ iff $T(x)$ reaches $\ell$, and let $z$ be the label of $\ell$. Then we have $C(\calD_z)=\E[A\mid\text{$T,\ell$ and $h(x)=z$}]=2\cdot\E[A\mid T,\ell]>0$ and similarly $C(\calD_{1-z})=2\cdot\E[B\mid T,\ell]$. Thus \[\varepsilon\cdot C(\calD_z)-(1-\varepsilon)\cdot C(\calD_{1-z})~=~2\cdot\bigl(\varepsilon\cdot\E[A\mid T,\ell]+(1-\varepsilon)\cdot\E[B\mid T,\ell]\,\bigr)~\ge~0.\]
\end{proof}

Now we work toward proving \autoref{thm:xor}. Throughout, $n$ is the input length of $\Xor$, and $m$ is the input length of $\GapMaj$. We have $\BPP(\GapMaj)\le 1$ by outputting the bit at a uniformly random position from the input. We describe one way of seeing that $\BPP_{1/n}(\GapMaj)\ge\PostBPP_{1/n}(\GapMaj)\ge\Omega(\log n)$ provided $m\ge\log n$. For $z\in\{0,1\}$, define $\calG_z$ as the uniform distribution over $\GapMaj^{-1}(z)$.

\begin{fact} \label{fact:gapmaj}
For every conjunction $C\colon\{0,1\}^m\to\{0,1\}$ of width $w\le m/7$ and for each $z\in\{0,1\}$, \[C(\calG_z)~\le~3^w\cdot C(\calG_{1-z}).\]
\end{fact}

\begin{proof}
By symmetry we just consider $z=0$. Suppose $C$ has $u$ positive literals and $v$ negative literals ($u+v=w$). Then \[\textstyle C(\calG_0)~=~\binom{m-w}{m/3-u}/\binom{m}{m/3}~\le~\binom{m-w}{m/3}/\binom{m}{m/3}~=~\frac{(2m/3)\cdot(2m/3-1)\cdots(2m/3-w+1)}{m\cdot(m-1)\cdots(m-w+1)}~\le~(2/3)^w,\] \begin{align*}\textstyle C(\calG_1)~=~\binom{m-w}{m/3-v}/\binom{m}{m/3}~\ge~\binom{m-w}{m/3-w}/\binom{m}{m/3}~&\textstyle=~\frac{(m/3)\cdot(m/3-1)\cdots(m/3-w+1)}{m\cdot(m-1)\cdots(m-w+1)}\\&\textstyle\ge~\bigl(\frac{m/3-w}{m-w}\bigr)^w~\ge~\bigl(\frac{m/3-m/7}{m-m/7}\bigr)^w~=~(2/9)^w.\end{align*} Thus $C(\calG_0)/C(\calG_1)\le\bigl(\frac{2/3}{2/9}\bigr)^w=3^w$.
\end{proof}

Combining \autoref{fact:postbpp} and \autoref{fact:gapmaj} (using $h=\GapMaj$, $k=m$, $\calD_z=\calG_z$, $\varepsilon=1/n$, and $w=\PostBPP_\varepsilon(h)$) implies that $(1-\varepsilon)/\varepsilon\le 3^w$, in other words we have $\PostBPP_{1/n}(\GapMaj)\ge\log_3(n(1-1/n))\ge\Omega(\log n)$, provided $w\le m/7$. If $w>m/7$ then $\PostBPP_{1/n}(\GapMaj)\ge\Omega(\log n)$ holds anyway provided $m\ge\log n$.

Hence, our result can be restated as follows.

\renewcommand{\customthmname}{Theorem \ref*{thm:xor} (Restated)}
\begin{customthm*} \label{thm:xor:restated}
$\PostBPP(\Xor\circ\GapMaj^n)\ge\Omega(n\log n)$ provided $m\ge\log n$.
\end{customthm*}

\begin{proof}
We show $\PostBPP(\Xor\circ\GapMaj^n)>\frac{1}{14}n\log n$. By \autoref{fact:postbpp} (using $h=\Xor\circ\GapMaj^n$, $k=nm$, and $\varepsilon=1/3$) it suffices to exhibit for each $z\in\{0,1\}$ a distribution $\calD_z$ over $(\Xor\circ\GapMaj^n)^{-1}(z)$, such that for every conjunction $C$ of width $\le\frac{1}{14}n\log n$ and for each $z\in\{0,1\}$, either $C(\calD_z)<2C(\calD_{1-z})$ or $C(\calD_z)=0$. Letting $\calF_z$ be the uniform distribution over $\Xor^{-1}(z)$, define $\calD_z$ as the mixture over $y\sim\calF_z$ of $\calG_y\coloneqq\calG_{y_1}\times\cdots\times\calG_{y_n}$ (i.e., $(x^1,\ldots,x^n)\sim\calG_y$ is sampled by independently sampling $x^i\sim\calG_{y_i}$ for all $i$). Put succinctly, $\calD_z\coloneqq\E_{y\sim\calF_z}[\calG_y]$. Letting $\calG\coloneqq\frac{1}{2}\calG_0+\frac{1}{2}\calG_1$ and $\calF\coloneqq\frac{1}{2}\calF_0+\frac{1}{2}\calF_1$ and $\calD\coloneqq\frac{1}{2}\calD_0+\frac{1}{2}\calD_1$, we have $\calD=\calG^n$ since $\calF$ is uniform over $\{0,1\}^n$. Since $C(\calD)=\frac{1}{2}C(\calD_0)+\frac{1}{2}C(\calD_1)$, our goal of showing ``$\frac{1}{2}C(\calD_0)<C(\calD_1)<2C(\calD_0)$ or $C(\calD_0)=C(\calD_1)=0$'' is equivalent to showing ``$\frac{2}{3}C(\calD)<C(\calD_1)<\frac{4}{3}C(\calD)$ or $C(\calD)=0$.''

Now consider any conjunction $C$ of width $w\le\frac{1}{14}n\log n$ such that $C(\calD)>0$, and write $C(x^1,\ldots,x^n)=\prod_iC_i(x^i)$ where $C_i$ is a conjunction. Since $C_i(\calG)=\frac{1}{2}C_i(\calG_0)+\frac{1}{2}C_i(\calG_1)$, for each $y_i\in\{0,1\}$ we can write $C_i(\calG_{y_i})=(1+a_i(-1)^{y_i})C_i(\calG)$ for some number $a_i$ with $|a_i|\le 1$ (so $a_i\ge 0$ iff $C_i(\calG_0)\ge C_i(\calG_1)$). Let $w_i$ be the width of $C_i$, so $\sum_iw_i=w\le\frac{1}{14}n\log n$. Then $w_i\le\frac{1}{7}\log n\le m/7$ for at least $n/2$ many values of $i$, and for such $i$ note that by \autoref{fact:gapmaj}, $C_i(\calG_{y_i})\le 3^{(\log n)/7}\cdot C_i(\calG_{1-y_i})\le n^{1/4}\cdot C_i(\calG_{1-y_i})$ for each $y_i\in\{0,1\}$. The latter implies that $|a_i|\le 1-2/(n^{1/4}+1)\le 1-n^{-1/4}$. Thus \[\textstyle\bigl|\prod_i a_i\bigr|~=~\prod_i|a_i|~\le~(1-n^{-1/4})^{n/2}~\le~e^{-n^{3/4}/2}~\le~1/4.\] For $S\subseteq[n]$, let $\chi_S\colon\{0,1\}^n\to\{1,-1\}$ be the character $\chi_S(y)\coloneqq\prod_{i\in S}(-1)^{y_i}=(-1)^{\sum_{i\in S}y_i}$. Note that $\E_{y\sim\calF_1}[\chi_S]$ is $1$ if $S=\emptyset$, is $-1$ if $S=[n]$, and is $0$ otherwise. Putting everything together,
\begin{align*}
\textstyle C(\calD_1)~&\textstyle=~\E_{y\sim\calF_1}[C(\calG_y)]~=~\E_{y\sim\calF_1}\bigl[\prod_iC_i(\calG_{y_i})\bigr]~=~\E_{y\sim\calF_1}\bigl[\prod_i(1+a_i(-1)^{y_i})C_i(\calG)\bigr]\\[5pt]
&\textstyle=~\bigl(\prod_iC_i(\calG)\bigr)\cdot\E_{y\sim\calF_1}\bigl[\sum_{S\subseteq[n]}\prod_{i\in S}a_i(-1)^{y_i}\bigr]~=~C(\calD)\cdot\sum_{S\subseteq[n]}\bigl(\prod_{i\in S}a_i\bigr)\cdot\E_{y\sim\calF_1}[\chi_S(y)]\\[5pt]
&\textstyle=~C(\calD)\cdot\bigl(1-\prod_{i\in[n]}a_i\bigr)~\in~C(\calD)\cdot(1\pm 1/4)
\end{align*}
which implies $\frac{2}{3}C(\calD)<C(\calD_1)<\frac{4}{3}C(\calD)$ since we are assuming $C(\calD)>0$. This concludes the proof of \hyperref[thm:xor:restated]{Theorem \ref*{thm:xor}}.
\end{proof}

Using strong LP duality (as in \cite{gavinsky14enroute}), it can be seen that \autoref{fact:postbpp} is a tight lower bound method up to constant factors: $\PostBPP_\varepsilon(h)\ge\Omega(c)$ iff it is possible to prove this via \autoref{fact:postbpp} by exhibiting ``hard input distributions'' $\calD_0$ and $\calD_1$ (as we did for $\GapMaj$ in \autoref{fact:gapmaj}). Since this was the only property of $g$ used in the proof of \hyperref[thm:xor:restated]{Theorem \ref*{thm:xor}}, this implies that $\BPP(\Xor\circ g^n)\ge\PostBPP(\Xor\circ g^n)\ge\Omega(n\cdot\PostBPP_{1/n}(g))$ holds for all $g$, as we mentioned in \autoref{sec:intro:xor}.


\section{Proof of \autoref{thm:maj}: \texorpdfstring{$\Maj$}{MAJ} sometimes necessitates amplification} \label{sec:maj}

We first discuss a standard technique for proving randomized query complexity lower bounds, which will be useful in the proof of \autoref{thm:maj}. For any conjunction $C\colon\{0,1\}^k\to\{0,1\}$ and distribution $\calD$ over $\{0,1\}^k$, we write $C(\calD)\coloneqq\E_{x\sim\calD}[C(x)]=\Pr_{x\sim\calD}[C(x)=1]$. The number of literals in a conjunction is called its width.

\begin{fact} \label{fact:wapp}
Let $h\colon\{0,1\}^k\to\{0,1\}$ be a partial function, and let $\calD_0$, $\calD_1$, $\calD_2$ be three distributions, over $h^{-1}(0)$, $h^{-1}(1)$, and $h^{-1}(0)\cup h^{-1}(1)$ respectively. Then for every $0<\varepsilon\le 1/10$ there exists a conjunction $C$ of width $\WAPP_\varepsilon(h)$ such that $C(\calD_0)\le\delta\cdot C(\calD_1)$ and $C(\calD_2)\le(1+\delta)\cdot C(\calD_1)$ and $C(\calD_1)>0$, where $\delta\coloneqq 2\sqrt{\varepsilon}$.
\end{fact}

The key calculation underlying the proof of \autoref{fact:wapp} is encapsulated in the following:

\begin{fact} \label{fact:rv}
Let $P_0$, $P_1$, $P_2$ be three jointly distributed nonnegative random variables with $\E[P_1]>0$. For any $0<\varepsilon\le 1/10$, if $\E[P_0]\le\varepsilon$ and $\E[P_1]\ge 1-\varepsilon$ and $\E[P_2]\le 1$, then there exists an outcome $o$ such that $P_0(o)\le\delta\cdot P_1(o)$ and $P_2(o)\le(1+\delta)\cdot P_1(o)$ and $P_1(o)>0$, where $\delta\coloneqq 2\sqrt{\varepsilon}$.
\end{fact}

\begin{proof}[Proof of \autoref{fact:rv}]
Let $W\coloneqq\{o\,:\,P_1(o)>0\}\ne\emptyset$. Suppose for contradiction that for every outcome $o\in W$, either $P_0(o)>\delta\cdot P_1(o)$ or $P_2(o)>(1+\delta)\cdot P_1(o)$. Then $W$ can be partitioned into events $U$ and $V$ such that $P_0(o)>\delta\cdot P_1(o)$ for every $o\in U$ and $P_2(o)>(1+\delta)\cdot P_1(o)$ for every $o\in V$. Letting $I_U$ and $I_V$ be the indicator random variables for these events, we have $\E[P_1\cdot I_U]+\E[P_1\cdot I_V]=\E[P_1]$ and thus either:
\begin{mylist}{2}
\item $\E[P_1\cdot I_U]\ge\sqrt{\varepsilon}\cdot\E[P_1]$, in which case \[\E[P_0]~\ge~\E[P_0\cdot I_U]~>~\delta\cdot\E[P_1\cdot I_U]~\ge~\delta\cdot\sqrt{\varepsilon}\cdot(1-\varepsilon)~=~2\varepsilon(1-\varepsilon)~>~\varepsilon,~~\text{or}\]
\item $\E[P_1\cdot I_V]\ge(1-\sqrt{\varepsilon})\cdot\E[P_1]$, in which case \[\E[P_2]~\ge~\E[P_2\cdot I_V]~>~(1+\delta)\cdot\E[P_1\cdot I_V]~\ge~(1+\delta)\cdot(1-\sqrt{\varepsilon})\cdot(1-\varepsilon)~>~1\] where the last inequality can be verified by a little calculus for $0<\varepsilon\le 1/10$.
\end{mylist}
Both cases yield a contradiction.
\end{proof}

\begin{proof}[Proof of \autoref{fact:wapp}]
Abbreviate $\WAPP_\varepsilon(h)$ as $r$. Fix a randomized decision tree of cost $r$ computing $h$ with error parameter $\varepsilon$ and threshold $t>0$ (from the definition of $\WAPP$), and assume w.l.o.g.\ that for each outcome of the randomness, the corresponding deterministic tree is a perfect tree with $2^r$ leaves, all at depth $r$. Consider the probability space where we sample a deterministic decision tree $T$ as an outcome of the randomized decision tree, and sample a uniformly random leaf $\ell$ of $T$. For any outcome $T,\ell$, let $C_{T,\ell}$ be the conjunction of width $r$ such that $C_{T,\ell}(x)=1$ iff $T(x)$ reaches $\ell$. Define three joint random variables $P_0$, $P_1$, $P_2$ as \[P_j(T,\ell)~\coloneqq~\begin{cases}C_{T,\ell}(\calD_j)&\text{if the label of $\ell$ is $1$}\\0&\text{if the label of $\ell$ is $0$}\end{cases}.\] Conditioned on any particular $x$ and $T$, the probability that $\ell$ is the leaf reached by $T(x)$ is $2^{-r}$. Thus
\begin{align*}
\textstyle\E[P_j]~&\textstyle=~\Pr_{T,\ell,\,x\sim\calD_j}[\text{$\ell$ is the leaf reached by $T(x)$ and its label is $1$}]\\
&\textstyle=~\E_{x\sim\calD_j}\bigl[2^{-r}\cdot\Pr_T[\text{$T(x)$ outputs $1$}]\bigr]
\end{align*}
which implies $\E[P_0]\le 2^{-r}t\varepsilon$ and $\E[P_1]\ge 2^{-r}t(1-\varepsilon)$ and $\E[P_2]\le 2^{-r}t$. Applying \autoref{fact:rv} to the scaled random variables $(2^r/t)P_0$, $(2^r/t)P_1$, $(2^r/t)P_2$ yields an outcome $T,\ell$ such that \[P_0(T,\ell)~\le~\delta\cdot P_1(T,\ell)\quad\text{and}\quad P_2(T,\ell)~\le~(1+\delta)\cdot P_1(T,\ell)\quad\text{and}\quad P_1(T,\ell)~>~0.\] Since $P_1(T,\ell)>0$, the label of $\ell$ must be $1$, so we get \[C_{T,\ell}(\calD_0)~\le~\delta\cdot C_{T,\ell}(\calD_1)\quad\text{and}\quad C_{T,\ell}(\calD_2)~\le~(1+\delta)\cdot C_{T,\ell}(\calD_1)\quad\text{and}\quad C_{T,\ell}(\calD_1)~>~0.\]
\end{proof}

Now we work toward proving \autoref{thm:maj}. Throughout, $n$ is the input length of $\Maj$, and $m$ is the input length of $\GapOr$. We have $\RP(\GapOr)\le 1$ by outputting the bit at a uniformly random position from the input. We describe one way of seeing that $\BPP_{1/n}(\GapOr)\ge\WAPP_{1/n}(\overline{\GapOr})\ge\Omega(\log n)$ provided $m\ge\log n$ (this cannot be shown via \autoref{fact:postbpp}). For $z\in\{0,1\}$, define $\calG_z$ as the uniform distribution over $\GapOr^{-1}(z)$.

\begin{fact} \label{fact:gapor}
For every conjunction $C\colon\{0,1\}^m\to\{0,1\}$:
\begin{mylist}{2}
\item[\i] $C(\calG_0)\in\{0,1\}$.
\item[\ii] If $C(\calG_0)=1$ and $C$ has width $w\le m/4$ then $C(\calG_1)\ge 3^{-w}$.
\end{mylist}
\end{fact}

\begin{proof}
\i: Note that $\calG_0$ is supported entirely on the input $0^m$. If $C$ has a positive literal then $C(\calG_0)=0$. If $C$ has only negative literals then $C(\calG_0)=1$.

\ii: Suppose $C$ has $w$ negative literals and no positive literals. Then \[\textstyle C(\calG_1)~=~\binom{m-w}{m/2}/\binom{m}{m/2}~=~\frac{(m/2)\cdot(m/2-1)\cdots(m/2-w+1)}{m\cdot(m-1)\cdots(m-w+1)}~\ge~\bigl(\frac{m/2-w}{m-w}\bigr)^w~\ge~\bigl(\frac{m/2-m/4}{m-m/4}\bigr)^w~=~3^{-w}.\]
\end{proof}

Combining \autoref{fact:wapp} and \autoref{fact:gapor} (using $h=\overline{\GapOr}$, $k=m$, $\calD_0=\calG_1$, $\calD_1=\calG_0$, $\calD_2$ is not needed, $\varepsilon=1/n$, and $w=\WAPP_\varepsilon(h)$) implies that $3^{-w}\le\delta$, in other words $\WAPP_{1/n}(\overline{\GapOr})\ge\log_3(1/(2\sqrt{1/n}))\ge\Omega(\log n)$, provided $w\le m/4$. If $w>m/4$ then $\WAPP_{1/n}(\overline{\GapOr})\ge\Omega(\log n)$ holds anyway provided $m\ge\log n$.

Hence, our result can be restated as follows.\footnote{Properties \i and \ii from \autoref{fact:gapor} are somewhat stronger than necessary for the proof of \autoref{thm:maj} to go through. The proof works, with virtually no modification, for any $g$ satisfying the following for some distributions $\calG_z$ over $g^{-1}(z)$ ($z\in\{0,1\}$): For every conjunction $C\colon\{0,1\}^m\to\{0,1\}$ such that $C(\calG_0)>0$, we have $C(\calG_1)\le C(\calG_0)$ and if furthermore $C$ has width $w\le m/4$ then $C(\calG_1)\ge 2^{-O(w)}\cdot C(\calG_0)$.}

\renewcommand{\customthmname}{Theorem \ref*{thm:maj} (Restated)}
\begin{customthm*} \label{thm:maj:restated}
$\WAPP_\varepsilon(\Maj\circ\GapOr^n)\ge\Omega(n\log n)$ for some constant $\varepsilon>0$ provided $m\ge\log n$.
\end{customthm*}

We show $\WAPP_{1/36}(\Maj\circ\GapOr^n)>\frac{1}{16}n\log n$. By \autoref{fact:wapp} (using $h=\Maj\circ\GapOr^n$, $k=nm$, $\varepsilon=1/36$, and $\delta=1/3$) it suffices to exhibit distributions $\calD_0$, $\calD_1$, $\calD_2$ over $h^{-1}(0)$, $h^{-1}(1)$, and $h^{-1}(0)\cup h^{-1}(1)$ respectively, such that for every conjunction $C$ of width $\le\frac{1}{16}n\log n$, either $C(\calD_0)>\frac{1}{3}C(\calD_1)$ or $C(\calD_2)>\frac{4}{3}C(\calD_1)$ or $C(\calD_1)=0$. Assume $n$ is even and for the tiebreaker, $\Maj(y)=1$ if $|y|=n/2$. For $\zeta\in\{0,1,2\}$ letting $\calF_\zeta$ be the uniform distribution over all $y\in\{0,1\}^n$ with $|y|=n/2-1+\zeta$ (so $\calF_0$, $\calF_1$, $\calF_2$ are over $\Maj^{-1}(0)$, $\Maj^{-1}(1)$, $\Maj^{-1}(1)$ respectively), define $\calD_\zeta$ as the mixture over $y\sim\calF_\zeta$ of $\calG_y\coloneqq\calG_{y_1}\times\cdots\times\calG_{y_n}$ (i.e., $(x^1,\ldots,x^n)\sim\calG_y$ is sampled by independently sampling $x^i\sim\calG_{y_i}$ for all $i$). Put succinctly, $\calD_\zeta\coloneqq\E_{y\sim\calF_\zeta}[\calG_y]$.

Now consider any conjunction $C$ of width $w\le\frac{1}{16}n\log n$, and write $C(x^1,\ldots,x^n)=\prod_iC_i(x^i)$ where $C_i$ is a conjunction. By \autoref{fact:gapor}.\i, $[n]$ can be partitioned into $A\cup B$ such that $C_i(\calG_0)=1$ for all $i\in A$, and $C_i(\calG_0)=0$ for all $i\in B$. Abbreviate $C_i(\calG_1)$ as $c_i$, and for $S\subseteq[n]$ write $c_S\coloneqq\prod_{i\in S}c_i$. Identify $y\in\{0,1\}^n$ with $Y\coloneqq\{i\,:\,y_i=1\}$, so $|y|=|Y|$. Let the uniform distribution over all size-$s$ subsets of $S$ be denoted by $\binom{S}{s}$, so $y\sim\calF_\zeta$ corresponds to $Y\sim\binom{[n]}{n/2-1+\zeta}$. Let $I_{Y\supseteq B}\coloneqq\prod_{i\not\in Y}C_i(\calG_0)$ be the indicator random variable for the event $Y\supseteq B$. Now for $\zeta\in\{0,1,2\}$,
\begin{align*}
\textstyle C(\calD_\zeta)~&\textstyle=~\E_{y\sim\calF_\zeta}[C(\calG_y)]~=~\E_{y\sim\calF_\zeta}\bigl[\prod_iC_i(\calG_{y_i})\bigr]~=~\E_{Y\sim\binom{[n]}{n/2-1+\zeta}}\bigl[c_Y\cdot I_{Y\supseteq B}\bigr]\\[5pt]
&\textstyle=~\underbrace{\Pr_{Y\sim\binom{[n]}{n/2-1+\zeta}}[Y\supseteq B]}_{\mbox{$p_\zeta$}}\,\cdot~c_B\,\cdot\,\underbrace{\E_{S\sim\binom{A}{n/2-1+\zeta-|B|}}[c_S]}_{\mbox{$q_\zeta$}}.
\end{align*}
If $c_B=0$ then $C(\calD_1)=0$, so assume $c_B>0$. Factoring out $c_B$ and defining $p_\zeta$ and $q_\zeta$ as above (but $q_\zeta$ is undefined if $p_\zeta=0$), our goal is to show that either $p_0q_0>\frac{1}{3}p_1q_1$ or $p_2q_2>\frac{4}{3}p_1q_1$ or $p_1q_1=0$. There are three cases depending on whether $|B|$ is greater than, equal to, or less than $n/2$. First we collect some generally useful properties:

\begin{claim} \label{clm:useful}
~\i~~$p_0=\frac{n/2-|B|}{n/2}\cdot p_1$~~and~~$p_1=\frac{n/2+1-|B|}{n/2+1}\cdot p_2$.~~~~\ii~~$0<q_1\le\sqrt{n}\cdot q_2$~~if $q_1$ is defined.
\end{claim}

\begin{proof}
\i: We just consider $p_0$ vs.\ $p_1$ since $p_1$ vs.\ $p_2$ is similar. Imagine sampling $Y_1\sim\binom{[n]}{n/2}$ and then obtaining the set $Y_0$ by removing a uniformly random $i\in Y_1$. If $Y_1\supseteq B$, then $Y_0\supseteq B$ when $i\in Y_1\smallsetminus B$, which happens with probability $\frac{n/2-|B|}{n/2}$ (assuming $|B|\le n/2$; if $|B|>n/2$ then $p_0=p_1=0$). Thus \[\textstyle p_0~=~\Pr[Y_0\supseteq B]~=~\Pr[Y_0\supseteq B\mid Y_1\supseteq B]\cdot\Pr[Y_1\supseteq B]~=~\frac{n/2-|B|}{n/2}\cdot p_1.\]

\ii: Let $w_i$ be the width of $C_i$, so $\sum_iw_i=w\le\frac{1}{16}n\log n$. Then $w_i\le\frac{1}{4}\log n\le m/4$ for at least $3n/4$ many values of $i$, and for such $i$ note that by \autoref{fact:gapor}.\ii, $c_i\ge 3^{-(\log n)/4}\ge n^{-2/5}$ if $i\in A$. This implies that if we sample a uniformly random $i$ from any $A'\subseteq A$ with $|A'|=n/2$ (note that $|A|\ge n/2$ if $q_1$ is defined) then $\E_{i\in A'}[c_i]\ge\frac{1}{2}\cdot n^{-2/5}+\frac{1}{2}\cdot 0\ge 1/\sqrt{n}$. Now to relate $q_2$ and $q_1$, \[q_2~=~\E_{S\sim\binom{A}{n/2-|B|}}\bigl[c_S\cdot\E_{i\in A\smallsetminus S}[c_i]\bigr]~\ge~\E_{S\sim\binom{A}{n/2-|B|}}\bigl[c_S/\sqrt{n}\bigr]~=~q_1/\sqrt{n}\] where the inequality uses $|A\smallsetminus S|=(n-|B|)-(n/2-|B|)=n/2$. Furthermore, $q_1>0$ if $q_1$ is defined, because $n/2-|B|\le|A|-n/4$ and thus there exists an $S\subseteq A$ with $|S|=n/2-|B|$ and $c_i\ge n^{-2/5}>0$ for all $i\in S$, hence $c_S>0$. (A similar argument shows $0<q_0\le\sqrt{n}\cdot q_1$ if $q_0$ is defined, but we will not need that.)
\end{proof}

\paragraph{Case $\bm{|B|>n/2}$.} In this case, $p_1=0$ so we are done.

\paragraph{Case $\bm{|B|=n/2}$.} By \autoref{clm:useful}, $p_2=p_1\cdot(n/2+1)$ and $q_2\ge q_1/\sqrt{n}>0$ and thus \[\textstyle p_2q_2~\ge~p_1q_1\cdot(n/2+1)/\sqrt{n}~>~\frac{4}{3}p_1q_1.\]

\paragraph{Case $\bm{|B|<n/2}$.} We will show that $\frac{p_0}{p_1}\ge\frac{1}{2}\cdot\frac{p_1}{p_2}$ and $\frac{q_2}{q_1}\ge\frac{9}{10}\cdot\frac{q_1}{q_0}$, which yields the punchline: \[\textstyle\text{If~\,$p_0q_0\le\frac{1}{3}p_1q_1$~\,then~\,$\frac{q_2}{q_1}~\ge~\frac{9}{10}\!\cdot\!\frac{q_1}{q_0}~\ge~\frac{9}{10}\!\cdot\! 3\!\cdot\!\frac{p_0}{p_1}~\ge~\frac{9}{10}\!\cdot\! 3\!\cdot\!\frac{1}{2}\!\cdot\!\frac{p_1}{p_2}~>~\frac{4}{3}\!\cdot\!\frac{p_1}{p_2}$~\,and thus~\,$p_2q_2>\frac{4}{3}p_1q_1$.}\] First, $\frac{p_0}{p_1}\ge\frac{1}{2}\cdot\frac{p_1}{p_2}$ follows from \autoref{clm:useful}.\i using $|B|\le n/2-1$: \[\textstyle\frac{p_0}{p_1}~=~\frac{n/2+1}{n/2}\cdot\frac{n/2-|B|}{n/2+1-|B|}\cdot\frac{p_1}{p_2}~\ge~1\cdot\frac{n/2-(n/2-1)}{n/2+1-(n/2-1)}\cdot\frac{p_1}{p_2}~=~\frac{1}{2}\cdot\frac{p_1}{p_2}.\] It just remains to show $\frac{q_2}{q_1}\ge\frac{9}{10}\cdot\frac{q_1}{q_0}$. Henceforth let $s\coloneqq n/2-1-|B|\ge 0$. The experiment $S\sim\binom{A}{s+2}$ in the definition of $q_2$ can alternatively be viewed as:
\begin{mylist}{0}
\item Sample $S_0\sim\binom{A}{s}$.
\item Sample $i\in A\smallsetminus S_0$ u.a.r.\ and let $S_1\coloneqq S_0\cup\{i\}$.
\item Sample $j\in A\smallsetminus S_1$ u.a.r.\ and let $S=S_2\coloneqq S_1\cup\{j\}$.
\end{mylist}
That is, $i$ and $j$ are sampled without replacement. We consider an ``ideal'' (easier to analyze) version of this experiment that samples $i$ and $j$ with replacement, in other words, the third step becomes:
\begin{mylist}{0}
\item Sample $j\in A\smallsetminus S_0$ u.a.r.\ and let $S_2^*\coloneqq S_1\cup\{j\}$.
\end{mylist}
Now $S_2^*$ is a \emph{multiset}, which may have two copies of $i$, in which case the product $c_{S_2^*}$ has two factors of $c_i$. Just as $q_2\coloneqq\E[c_{S_2}]$, we let $q_2^*\coloneqq\E[c_{S_2^*}]$, and we next show how to derive $\frac{q_2^*}{q_1}\ge\frac{q_1}{q_0}$ from the following claim:

\begin{claim} \label{clm:growth}
For all nonnegative numbers $\alpha_1,\ldots,\alpha_N$ and $\beta_1,\ldots,\beta_N$ such that $\alpha_k\beta_k>0$ for some $k$, \[\frac{\sum_k\alpha_k\beta_k^2}{\sum_k\alpha_k\beta_k}~\ge~\frac{\sum_k\alpha_k\beta_k}{\sum_k\alpha_k}.\]
\end{claim}

\begin{proof}
By clearing denominators, this inequality is equivalent to \[\textstyle\bigl(\sum_k\alpha_k\bigr)\bigl(\sum_k\alpha_k\beta_k^2\bigr)~\ge~\bigl(\sum_k\alpha_k\beta_k\bigr)^2\] which can be rewritten as \[\textstyle\sum_{k,\ell}\alpha_k\alpha_\ell\beta_\ell^2~\ge~\sum_{k,\ell}\alpha_k\beta_k\alpha_\ell\beta_\ell.\] Subtracting $\sum_k\alpha_k^2\beta_k^2$ from both sides, this is equivalent to \[\textstyle\sum_{k<\ell}\bigl(\alpha_k\alpha_\ell\beta_\ell^2+\alpha_\ell\alpha_k\beta_k^2\bigr)~\ge~\sum_{k<\ell}2\alpha_k\beta_k\alpha_\ell\beta_\ell.\] We show that this inequality holds for each summand separately. Factoring out $\alpha_k\alpha_\ell$, this reduces to showing $\beta_\ell^2+\beta_k^2\ge 2\beta_k\beta_\ell$, which holds since \[\textstyle\beta_\ell^2+\beta_k^2-2\beta_k\beta_\ell~=~(\beta_\ell-\beta_k)^2~\ge~0.\]
\end{proof}

In the statement of \autoref{clm:growth}, let the index $k$ correspond to $S_0$, let $N\coloneqq\binom{|A|}{s}$, let $\alpha_k\coloneqq c_{S_0}/N$, and let $\beta_k\coloneqq\E_{i\in A\smallsetminus S_0}[c_i]$. Then \[\textstyle q_0=\sum_k\alpha_k\quad\text{and}\quad q_1=\sum_k\alpha_k\beta_k\quad\text{and}\quad q_2^*=\sum_k\alpha_k\beta_k^2\] and $q_0\ge q_1>0$ by \autoref{clm:useful}.\ii (i.e., $\alpha_k\beta_k>0$ for some $k$) so by \autoref{clm:growth} we indeed have $\frac{q_2^*}{q_1}\ge\frac{q_1}{q_0}$. To conclude that $\frac{q_2}{q_1}\ge\frac{9}{10}\cdot\frac{q_1}{q_0}$, we just need to show $q_2\ge\frac{9}{10}q_2^*$.

The third step of the $S_2$ experiment is just the third step of the $S_2^*$ experiment conditioned on $j\ne i$, which happens with probability $1-\frac{1}{|A|-s}$. With probability $\frac{1}{|A|-s}$, we get $j=i$ in the $S_2^*$ experiment. If we condition on the latter event, it yields another experiment, whose result we call $S_2^\text{err}$, which is a multiset definitely containing two copies of $i$. Correspondingly we define $q_2^\text{err}\coloneqq\E[c_{S_2^\text{err}}]$ (with two factors of $c_i$). Now we have \[\textstyle q_2^*~=~\Pr[j\ne i]\cdot\E\bigl[c_{S_2^*}\bigmid j\ne i\bigr]+\Pr[j=i]\cdot\E\bigl[c_{S_2^*}\bigmid j=i\bigr]~=~\bigl(1-\frac{1}{|A|-s}\bigr)\cdot q_2+\frac{1}{|A|-s}\cdot q_2^\text{err}~\le~q_2+\frac{2}{n}\cdot q_2^\text{err}\] since $|A|-s=(n-|B|)-(n/2-1-|B|)=n/2+1\ge n/2$.

The $S_2^\text{err}$ experiment can alternatively be viewed as:
\begin{mylist}{0}
\item Sample $S_1\sim\binom{A}{s+1}$.
\item Sample $i\in S_1$ u.a.r.\ and let $S_2^\text{err}\coloneqq S_1\cup\{i\}$.
\end{mylist}
This implies that $q_2^\text{err}\le q_1$ because the extra factor of $c_i\le 1$ cannot increase the expectation. By \autoref{clm:useful}.\ii we get $q_2^\text{err}\le q_1\le\sqrt{n}\cdot q_2$. Combining, we have \[\textstyle q_2^*~\le~q_2+\frac{2}{n}\cdot\sqrt{n}\cdot q_2~=~\bigl(1+\frac{2}{\sqrt{n}}\bigr)q_2~\le~\frac{10}{9}q_2\] and thus $q_2\ge\frac{9}{10}q_2^*$ as desired. This concludes the proof of \hyperref[thm:maj:restated]{Theorem \ref*{thm:maj}}.


\section{Open questions} \label{sec:open}

\begin{open} \label{open:total-amp}
Is there a total function $g\colon\{0,1\}^m\to\{0,1\}$ such that $\BPP(\Xor\circ g^n)\ge\Omega(n\log n\cdot\BPP(g))$ or~\,$\BPP(\Maj\circ g^n)\ge\Omega(n\log n\cdot\BPP(g))$?
\end{open}

Since \autoref{fact:gapor} captures the only properties of $g=\GapOr$ used in our proof of \autoref{thm:maj}, this provides a possible roadmap for confirming \autoref{open:total-amp}: Just find a total function $g$ satisfying properties similar to \autoref{fact:gapor}, enabling our proof of \autoref{thm:maj} to go through. However, such a $g$ would need to have certificate complexity $\omega(\BPP(g))$, and it remains a significant open problem to find any such total function $g$ (the ``pointer function'' \cite{goos18deterministic,ambainis17separations} and ``cheat sheet'' \cite{aaronson16separations} methods do not seem to work).

Another approach for confirming \autoref{open:total-amp} would be to generalize the strong direct sum theorem from \cite{blais19optimal} to show that $\BPP(\Xor\circ g^n)\ge\Omega(n\cdot\BPPbar_{1/n}(g))$ or $\BPP(\Maj\circ g^n)\ge\Omega(n\cdot\BPPbar_{1/n}(g))$ holds for all $g$. This would answer \autoref{open:total-amp} in the affirmative, since \cite{blais19optimal} designed a total function $g$ satisfying $\BPPbar_{1/n}(g)\ge\Omega(\RP(g)\cdot\log n)$ using the ``pointer function'' method. Compared to our approach from the previous paragraph, this approach involves less stringent requirements on $g$, which makes it easier to design $g$ but harder to prove the composition lower bound.

\begin{open} \label{open:total-side}
Is there a total function $f\colon\{0,1\}^n\to\{0,1\}$ such that $\BPP^*(f)\ge\omega(\BPP^\dagger(f))$ (or similarly, $\BPP(f\circ\GapMaj^n)\ge\omega(\BPP(f\circ\GapOr^n))$)?
\end{open}

It is not difficult to find such a \emph{partial} function $f$. Namely, take any function $f'\colon\{0,1\}^n\to\{0,1\}$ such that $\BPP^*(f')\ge\Omega(n\log n)$, such as $f'=\Xor$ or $f'=\Maj$. Then take $f=f'\circ\Which^n$, which has input length $2n$ (recall from \autoref{sec:intro:formulate} that given $y\in\{0,1\}^2$ with the promise that $y$ has Hamming weight $1$, $\Which(y)$ indicates the location of the unique $1$ in $y$). A simple reduction shows $\BPP^*(f)\ge\BPP^*(f')$. However, $\BPP^\dagger(f)\le O(n)$: For each block of $2$ bits, we can repeatedly query both until one of them returns $1$ (which takes $O(1)$ queries in expectation). After doing this for all $n$ blocks (which takes $O(n)$ queries in expectation), we know for sure what the entire actual input is. By Markov's inequality, we can abort the execution after $O(n)$ queries while introducing only a small constant error probability. (Intuitively, composition with $\Which$ preserves hardness for $2$-sided noise but converts $1$-sided noise to ``$0$-sided noise'', and no partial function needs $\omega(n)$ queries in the setting of $0$-sided noise.)
\bigskip

In communication (rather than query) complexity, somewhat analogous questions have been studied in specific contexts \cite{molinaro13beating,blais14information,saglam18near}. The proof of \autoref{thm:or} also works for communication complexity. It would be interesting to develop analogues of \autoref{thm:xor} and \autoref{thm:maj} for communication complexity.


\appendix
\section{Proof of \autoref{thm:or}: \texorpdfstring{$\Or$}{Or} never necessitates amplification} \label{sec:or}

For completeness, we provide a self-contained proof that $\BPP^*(\Or)\le O(n)$, using the following standard fact about random walks (``the drunkard at the cliff'').

\begin{lemma} \label{lem:random-walk}
Consider a random walk on the integers that begins at $0$ and in each step moves right $(+1)$ with probability $p$ and moves left $(-1)$ with probability $1-p$.
\begin{mylist}{2}
\item[\i] If $p<1/2$ then the expected time at which the walk first visits $-1$ is $1/(1-2p)$.
\item[\ii] If $p>1/2$ then the probability that the walk ever visits $-1$ is $(1-p)/p$.
\end{mylist}
\end{lemma}

\begin{proof}[Proof of \autoref{lem:random-walk}]
\i: If random variable $X$ represents the time at which the walk first visits $-1$, then its expectation satisfies $\E[X]=1+p\cdot2\E[X]$ since after the first step, it either is already at $-1$, or is at $+1$ in which case to reach $-1$ it must first get back to $0$ ($\E[X]$ expected time) then from there get to $-1$ (another $\E[X]$ expected time). This equation has a unique solution $\E[X]=1/(1-2p)<\infty$.

\ii: If event $E$ represents the walk ever visiting $-1$, then its probability satisfies $\Pr[E]=(1-p)\cdot 1+p\cdot\Pr[E]^2$ since after the first step, it either is already at $-1$, or is at $+1$ in which case to reach $-1$ it must first get back to $0$ (probability $\Pr[E]$) then from there get to $-1$ (again probability $\Pr[E]$). This equation has two solutions $\Pr[E]\in\{(1-p)/p,1\}$. To rule out $\Pr[E]=1$, we define $q_k$ as the probability that the walk visits $-1$ within the first $k$ steps, and we show by induction on $k$ that $q_k\le(1-p)/p$. The base case is trivial since $q_0=0$. Assuming $q_k\le(1-p)/p$ we show $q_{k+1}\le(1-p)/p$. After the first step, with probability $1-p$ it is already at $-1$, and with probability $p$ it is at $+1$. In the latter case, to get to $-1$ within a total of $k+1$ steps (including the first step), it must get from $+1$ to $0$ and then from there it must get to $-1$, all within $k$ more steps; in particular, the walk must get from $+1$ to $0$ within $k$ steps (probability $\le q_k$) and then from $0$ to $-1$ within $k$ steps (probability $\le q_k$). Overall we can bound $q_{k+1}\le(1-p)\cdot 1+p\cdot q_k^2\le(1-p)+p\cdot(1-p)^2/p^2=(1-p)/p$.
\end{proof}

\begin{proof}[Proof of \autoref{thm:or}]
We may assume the noise probabilities are $\le 1/4$ (rather than just $\le 1/3$), because whenever an input bit is queried, we can instead query it five times and pretend that the majority vote was the result of the single query. This would only affect the cost by a constant factor. With this assumption, here is our decision tree, on input $y\in\{0,1\}^n$:
\bigskip

For $i=1,2,\ldots,n$:

\hspace*{12pt}Repeat:

\hspace*{24pt}Query $y_i$.

\hspace*{24pt}If the queries to $y_i$ have resulted in more $0$s than $1$s so far,

\hspace*{36pt}then break out of the inner loop.

\hspace*{24pt}If a total of $6n$ queries have been made (across all input bits), then halt and output $1$.

Halt and output $0$.
\bigskip

\noindent This decision tree's cost is $\le 6n$. To see the correctness, consider any input $y\in\{0,1\}^n$ and any tuple of noise probabilities ($\nu_1,\ldots,\nu_n)$ where each $\nu_i\le 1/4$. For each $i$, the random variable \[\text{``number of $1$s minus number of $0$s, among the queries to $y_i$ so far''}\] is a random walk with move-right probability $p_i=\nu_i\le 1/4$ if $y_i=0$ and $p_i=1-\nu_i\ge 3/4$ if $y_i=1$, and which stops when it visits $-1$.

First assume $\Or(y)=0$. Then for each $i$, $y_i=0$ and so by \autoref{lem:random-walk}.\i, the expected number of queries until the inner loop is broken is $1/(1-2p_i)\le 2$. By linearity, the expected total number of queries until all $n$ inner loops have been broken is $\le 2n$, so by Markov's inequality this number of queries is $<6n$ with probability $\ge 2/3$. Thus the decision tree outputs $0$ with probability $\ge 2/3$.

Now assume $\Or(y)=1$. Then for some $i$, $y_i=1$ and so by \autoref{lem:random-walk}.\ii, with probability $1-(1-p_i)/p_i=2-1/p_i\ge 2/3$ there would never be more $0$s than $1$s from the queries to $y_i$. In that case, the decision tree would never break out of the $i^\text{th}$ inner loop, even if it were allowed to run forever. Thus the decision tree outputs $1$ with probability $\ge 2/3$.
\end{proof}


\medskip
\subsection*{Acknowledgments}

We thank Badih Ghazi for interesting discussions about this work, and we thank anonymous reviewers for their comments. T.\ Watson was supported by NSF grant CCF-1657377.


\DeclareUrlCommand{\Doi}{\urlstyle{sf}}
\renewcommand{\path}[1]{\small\Doi{#1}}
\renewcommand{\url}[1]{\href{#1}{\small\Doi{#1}}}
\bibliographystyle{alphaurl}
\bibliography{log}

\end{document}